\documentclass[submission, nomarks]{dmtcs-episciences}
\usepackage{graphicx}

\usepackage[utf8]{inputenc}
\usepackage[english]{babel}

\usepackage{hyperref}

\usepackage{relsize}
\usepackage{amsfonts,amsmath,amssymb,amsbsy,amsthm,enumerate}
\usepackage[autostyle]{csquotes}

\usepackage[shortlabels]{enumitem}
\usepackage{comment}
\usepackage{placeins}

\usepackage{subfigure}

\usepackage{setspace}

\usepackage{cleveref}

\usepackage{algorithm}
\usepackage[noend]{algpseudocode}
\usepackage{multicol}

\newcommand{\ZZ}{\mathbb{Z}}
\newcommand{\NN}{\mathbb{N}}

\newcommand{\VV}{\mathcal{V}}
\newcommand{\FF}{\mathcal{F}}

\newcommand{\DD}{\mathcal{D}}
\newcommand{\SSS}{\mathcal{S}}
\newcommand{\RR}{\mathrm{R}}

\newtheorem{thm}{Theorem}

\newtheorem{claim}{Claim}

\newtheorem{prop}{Proposition}
\newtheorem{ppty}{Property}
\newtheorem{conj}{Conjecture}

\theoremstyle{remark}

\newtheorem{deff}{Definition}
\theoremstyle{definition}

\makeatletter
\def\BState{\State\hskip-\ALG@thistlm}
\makeatother

\title{The 2-domination and Roman domination numbers of grid graphs}
\author{Micha\"el Rao \and  Alexandre Talon}

\affiliation{ENSL, Univ Lyon, UCBL, LIP, MC2, F-69342, LYON Cedex 07, France}

\begin{document}
\publicationdetails{21}{2019}{1}{9}{4952}
\received{2018-11-02}
\revised{2019-04-17}
\accepted{2019-5-14}

\maketitle
\begin{abstract}
We investigate the 2-domination number for grid graphs, that is the size of a smallest set $D$ of vertices of the grid such that each vertex of the grid belongs to $D$ or has at least two neighbours in $D$. We give a closed formula giving the 2-domination number of any $n \!\times\! m$ grid, hereby confirming the results found by Lu and Xu, and Shaheen et al. for $n \leq 4$ and slightly correct the value of Shaheen et al. for $n = 5$. The proof relies on some dynamic programming algorithms, using transfer matrices in (min,+)-algebra. We also apply the method to solve the Roman domination problem on grid graphs.
\end{abstract}

\section{Introduction and notations}
A \emph{dominating set} $D$ in a graph $G$ is a subset of the vertices such that every vertex in $V(G) \setminus D$ has at least one neighbour in $D$. The \emph{domination number} of $G$, denoted by $\gamma(G)$ is the minimum size of a dominating set of $G$.
In 1993 Chang \cite{chang} conjectured that the domination number for a grid graph of arbitrary size was $\gamma(G_{n,m}) = \left\lceil \frac{(n+2)(m+2)}{5}\right\rceil-4$. He also showed that this was actually an upper bound. Gonçalves et al. \cite{rao} proved Chang's conjecture in 2011 by showing that this was a lower bound.
Several generalisations of the domination problem have also been studied in the literature (see for example \cite{bon}).
We adapt here the method used in \cite{rao} to the 2-domination problem and the Roman domination, obtaining closed formulas for these two problems on grids. This confirms the results found by \cite{lu-xu, shaheen} for $n \leq 4$ and slightly corrects the result by \cite{shaheen} for $n = 5$.
This also confirms the results of Pavlič and Žerovnik \cite{pav} for $n\le 8$ and prove that the upper bound given by Currò \cite{curro} is tight.

The domination problem is one of many problems which are hard for general graphs, but are easy to solve for graphs of bounded treewidth. The grids are among the simplest graphs which neither have a bounded treewidth nor a bounded cliquewidth, but for which these kinds of problems are usually difficult to tackle. Gonçalves et al. managed to solve the domination problem on grids. In this paper we generalise their techniques and try to see on what kind of problems they can be used.
The domination problems are related to tiling problems. For instance, \autoref{domination-shape} shows the shape associated with the domination problem. A smallest dominating set in a grid is equivalent to a smallest covering set of the rectangle with this shape. The method of Gonçalves et al. works thanks to the fact that the shape has the following three properties. First, it can tile (that is, cover without overlaps) the infinite plane. Second, there is a unique way, up to mirrors and translations, to tile the plane with this shape. Third, we can find optimal solutions which consist in projecting a tiling of the plane, cropping it and modifying only tiles at bounded distance from the border. We will discuss these properties in the conclusion.

\begin{figure}[h!]
\centering
\includegraphics[scale=0.6]{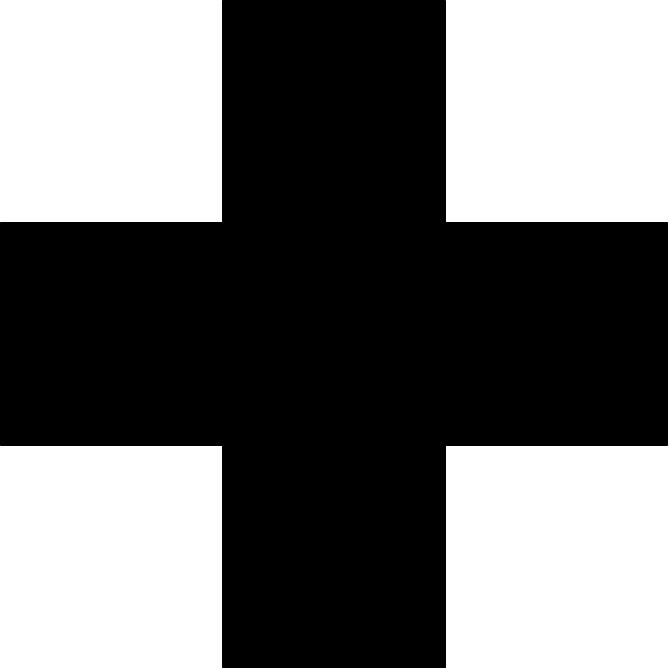}
\caption{The shape corresponding to the domination tiling problem.}
\label{domination-shape}
\end{figure}

We adapt here the method developed by Gonçalves et al. to some other domination problems: the 2-domination and the Roman domination. Each of these two problems is also related to a tiling problem, which also has the properties we have just mentioned.

The source of the program used to prove the results can be found in the arXiv version of this paper: \url{https://arxiv.org/format/1810.12896}, selecting the "Download source" option.

\section{Method of the proofs and application to the 2-domination problem}
 A \emph{2-dominating} set of $G$ is a subset $D \subset V(G)$ such that every vertex not in $D$ has at least two neighbours in $D$. The \emph{2-domination number} of a graph $G$, denoted by $\gamma_2(G)$ is the minimum size of a 2-dominating set of $G$. Here we compute all the $\gamma_2$ values for grid graphs. We denote by $G_{n,m}$ the grid graph with $n$ lines and $m$ columns, and by $\gamma_2(n,m)$ the 2-domination number for an $n\!\times\! m$ grid. We take the viewpoint of a grid whose cells represent vertices.
 In this section we give the tools to prove the following theorem.

\vspace{8cm}
\begin{thm}
For all $1 \leq n \leq m$,\\
\[\gamma_2(n,m) =    \left\{
\setstretch{1.25}
\begin{array}{ll}
      \left\lceil\frac{m+1}{2}\right\rceil & \quad\textrm{if }n = 1 \\
      m & \quad\textrm{if }n=2 \\
      m+\left\lceil\frac{m}{3} \right\rceil & \quad\textrm{if } n=3 \\
      2m-\left\lfloor \frac{m}{4} \right\rfloor& \quad\textrm{if } n=4\textrm{ and } m\mod 4 = 3 \\
      2m-\left\lfloor \frac{m}{4}+1 \right\rfloor& \quad\textrm{if } n=4\textrm{ and } m\mod 4 \neq 3 \\
      2m+\left\lceil\frac{m}{7}\right\rceil+1& \quad\textrm{if } n=5\textrm{ and } m\mod 7 \in \{0,6\}\\
      2m+\left\lceil\frac{m}{7}\right\rceil& \quad\textrm{if } n=5\textrm{ and } m\mod 7 \notin \{0,6\}\\
      2m+\left\lfloor\frac{6m}{11}\right\rfloor+1& \quad\textrm{if } n=6\textrm{ and } m\mod 11 \in \{0,2,6\}\\
      2m+\left\lfloor\frac{6m}{11}\right\rfloor+2& \quad\textrm{if } n=6 \textrm{ and } m\mod 11 \notin \{0,2,6\}\\
      3m-\left\lfloor\frac{m}{18}\right\rfloor+1& \quad\textrm{if } n=7 \textrm{ and } m > 9\textrm{ and }m\mod 18 \leq 9\\
      3m-\left\lfloor\frac{m}{18}\right\rfloor+1& \quad\textrm{if } n=7 \textrm{ and } (m \leq 9\textrm{ or }m\mod 18 > 9)\\
      3m+\left\lfloor\frac{m}{3}\right\rfloor& \quad\textrm{if } n=8 \textrm{ and } m \mod 3 = 1\\
      3m+\left\lfloor\frac{m}{3}\right\rfloor+1& \quad\textrm{if } n=8 \textrm{ and } m \mod 3 \neq 1\\
      \left\lfloor \frac{(n+2)(m+2)}{3}\right\rfloor-6 & \quad\textrm{if } n\geq 9.
\end{array}
\right. \]
\label{th-2dom}
\end{thm}

The first subsection uses some well-known techniques to establish the 2-domination values for grids of small height, whereas the second subsection uses the recent technique of loss, which was introduced by Gonçalves et al. \cite{rao} in order to obtain the values for arbitrarily large height for the domination problem in grids.

If $V$ is a vector of size $n$ indexed from 0 to $n-1$, we will sometimes refer, for concision, to some $V[i]$ where $i$ is negative or greater than $n-1$. This means that any condition or operation on invalid indices is to be ignored: for instance, "both $V[i-2] = 1$ and $V[i+1] = 1$" is necessarily false if $i < 2$ or $i > n-2$.
We denote by $|V|_p$ the number of entries of $V$ which have value $p$. 
Operations on matrices are always done in the $(\min,+)$-algebra.
In all what follows, unless explicitly specified otherwise, the grid we consider always has $n$ lines and $m$ columns. Note that we require $n \leq m$ in all theorems.
We denote by $a\mod b$ the remainder of the Euclidean division of $a$ by $b$.

\subsection{Computing 2-domination numbers for small $n$}
We use here a dynamic programming approach. Our algorithm is exponential in the number of lines, but for a fixed number of lines $n$, it is linear in the number of columns.

Let $\SSS = \{$ \textsc{stone}, \textsc{need\_one}, \textsc{ok} $\}$ be the set of cell states. If $D$ is a 2-dominating set of cells of the grid $G_{n,m}$ let $f(D) \in (\SSS^n)^m$ be such that $f(D)[i][j]$ is \textsc{stone} if $(i,j) \in D$, \textsc{ok} if at least two among $(i-1,j)$, $(i, j-1)$ and $(i, j+1)$ are in $D$, or \textsc{need\_one} otherwise. Note that the state of a cell does not depend on the values of the cells of the next columns. We also define, for $0 \leq i < n$, $f_i(D)$ to be the column $i$ of $f(D)$: for all $j\in\{0,\ldots,n-1\}$, $f_i(D)[j]=f(D)[i][j]$. Note that, since $D$ is 2-dominating, what precedes implies that $f_i(D)[j]$ = \textsc{need\_one} if exactly one among $(i-1,j)$, $(i, j-1)$ and $(i, j+1)$ is in $D$.

We now define the set of \emph{valid} states $\VV = \cup_{0\leq i < m}{\{f_i(D) : D \textrm{ is a 2-dominating set}\}}$. $\VV$ is composed of the states which we can find in some 2-dominating set: all states which can appear. Among these, we define the set of \emph{first} states $\FF = \{f_0(D) : D \textrm{ is a 2-dominating set}\}$. Finally, we define the set of \emph{dominated} states $\DD = \{f_{m-1}(D) : D \textrm{ is a 2-dominating set}\}$. $\FF$ is the set of states which can be the first column of a 2-dominating set, that is whose entries only depend on themselves and not on a previous column. $\DD$ is the list of states which do not need a next column to be 2-dominated: they are dominated by themselves and their previous column.

We now define the relation of \emph{compatibility} $\RR$: we say that a state $S'\in \VV$ is compatible with $S\in \VV$, and write $S\RR S'$ if there exist a 2-dominating set $D$ and $i \in \{0, \ldots, m-2\}$ such that $\textrm{f}_i(D) = S$ and $\textrm{f}_{i+1}(D) = S'$. 
 Defining these states enables us to use the principles of dynamic programming: instead of enumerating all possible 2-dominating sets, we realise that the information conveyed in $f(D)$ is enough, and that we only need the information at a column $i$ to continue to column $i+1$. In particular, we do not need to know what happened in previous columns.

To illustrate these concepts, we give the rules defining the sets $\VV$, $\FF$, $\DD$ and the relation $\RR$. $S \in \VV$ if and only if for all $ i \in \{0, \ldots, n-1\}$:
\begin{itemize}[topsep=0pt, noitemsep]
\item if $S[i] = $\textsc{need\_one} then at most one among $S[i-1]$, $S[i+1]$ is \textsc{stone};
\item if $S[i] = $\textsc{ok} then at least one among $S[i-1]$ and $S[i+1]$ is \textsc{stone}.
\end{itemize}
\vspace{5pt}
$S \in \FF$ if and only if for all $i \in \{0, \ldots, n-1\}$:
\begin{itemize}[topsep=0pt]
\setlength{\itemsep}{0pt}
\item if $S[i] = $\textsc{need\_one} then exactly one among $S[i-1]$, $S[i+1]$ is \textsc{stone};
\item if $S[i] = $\textsc{ok} then both $S[i-1]$ and $S[i+1]$ are \textsc{stone} (so $1 \leq i < m-1$);
\end{itemize}
\vspace{5pt}
A state $S$ belongs to $\DD$ if and only if $S\in \VV$ and none of its entries is \textsc{need\_one}.

Finally, $S \RR S'$ if and only if for all $i \in \{0, \ldots, n-1\}$:
\begin{itemize}
\setlength{\itemsep}{0pt}
\item if $S[i] = $\textsc{need\_one} then $S'[i] = $\textsc{stone};
\item if $S'[i] = $\textsc{need\_one} then exactly one among $S'[i-1]$, $S'[i+1]$ and $S[i]$ is \textsc{stone};
\item if $S'[i] = $\textsc{ok} then at least two among $S'[i-1]$, $S'[i+1]$ and $S[i]$ are \textsc{stone};
\end{itemize}

\begin{claim}
Let $F$ be the vector of size $|\VV|$ such that $F[S] = |S|_{\text{\textsc{stone}}}$ if $S\in \FF$ or $+\infty$ otherwise. Let $D$ be the vector of size $|\VV|$ such that $D[S] = 0$ if $S \in \DD$ or $+\infty$ otherwise. Let $T$ be the square matrix with $|\VV|$ lines such that $T[S][S'] = |S'|_{\text{\textsc{stone}}}$ if $S \RR S'$ or $+\infty$ otherwise.\\Then $\gamma_2(n,m) = F^{\mathsmaller T}T^{m-1}D$.\\
(We recall that the products of matrices are done in the $(\min, +)$-algebra.)
\label{claim-exact}
\end{claim}

\begin{proof}
Let $m \geq 1$. $V = F^{\mathsmaller T}T^{m-1}$ is a vector such that if $S \in \VV$ then $V[S]$ is the minimum size of a set $X$ which 2-dominates the subgrid with $n$ lines and $m-1$ columns, and such that $f(X)[m-1] = S$. However, we are interested in a 2-dominating set, therefore $S$ should be 2-dominated as well. Thus \hspace{4cm}$VD = \min_{S \in \DD}{V[S]}$ gives us the minimum size of any 2-dominating set.
\end{proof}

This claim leads to a simple algorithm to generate the different sets and the compatibility relation, and then compute the product and exponentiation of matrices, and a matrix vector product. The matrix $T$ is a transfer matrix, whose exponentiation propagates the fact of being 2-dominated one column further.
This is enough to compute the 2-domination numbers for small $n$ and $m$, but our goal is to find all the numbers for small $n$ and arbitrary $m$. What follows fills this hole.\\

We say that a matrix $M$ is \emph{primitive} if there exists an integer $k > 0$ such that $\max(M^k) < +\infty$.
\begin{claim}
$T$ is primitive.
\label{claim-primitive}
\end{claim}

\begin{proof}
Let $S_0, S_2 \in \VV$. Let $S_*$ be the state whose entries all are \textsc{stone}. There exists some $S_1 \in \VV$ such that $S_0 \RR S_*$, $S_* \RR S_1$ and $S_1 \RR S_2$. We leave the construction of $S_1$ to the reader: put the necessary stones and fill the rest accordingly. We conclude that $T^3 < +\infty$.
\end{proof}

\begin{prop}
Let $M$ be a primitive matrix of size $n$ and $k$ be such that $M^k < +\infty$. Let $V$ be a vector of size $n$ which has at least one entry different from $+\infty$. We put $V_i = M^iV$. Then there exist some $i_0, p$ and $r$ such that for all $i \geq i_0$, $V_i = V_{i-r}+p$.
\label{prop-primitive}
\end{prop}

\begin{proof}
Since $M^k < +\infty$, $V_k = M^kV$ has at least one entry different from $+\infty$. This implies that $V_{k-1}$ has one entry different from $+\infty$. Since $V_k = M^kV = MV_{k-1}$ we can prove by induction that all the $V_i$'s share this property.
Let $i \geq k$, and $b_0$ be such that $V_{i-k}[b_0]$ is a minimum of $V_{i-k}$. $V_i[a] \leq M^k[a][b_0]V_{i-k}[b_0] \leq \alpha + V_{i-k}[b_0]$, where $\alpha = \max(M^k)$.
Similarly, $V_i[a] \geq \min_{a', b'}{(M^k[a'][b'])} + V_{i-k}[b_0] \geq V_{i-k}[b_0]$.
 These two inequalities imply that $\max(V_i) - \min(V_i) \leq \alpha$ for every $i \geq k$. Let us write, for $i \geq k$, $V_i = n_i + V'_i$ where $n_i \in \NN$ and each entry of $V'_i$ is between $0$ and $\alpha$. Since there are at most $\alpha^n$ different vectors of size $n$ which have all entries between 0 and $\alpha$, there exist $k \leq i_1 < i_2$ such that $V'_{i_2} = V'_{i_1}$, hence by letting $r = n_{i_2} - n_{i_1}$ we obtain $V_{i_2} - V_{i_1} = r$.
\end{proof}

\autoref{prop-primitive} guarantees that the transfer matrix $T$ verifies, for some $p_0, k$ and $t$ the relation $T^{p+k} = T^p + t$ for $p \geq p_0$. Thanks to \autoref{claim-exact} the relations we obtain for the transfer matrix $T$ directly apply to the 2-domination number. Here are the relations we obtain for $n \leq 12$:
\begin{multicols}{2}
\begin{itemize}
\item $\forall m \geq 3, \gamma_2(1,m) = \gamma_2(1,m-2) + 1$;
\item $\forall m \geq 3, \gamma_2(2,m) = \gamma_2(2,m-1) + 1$;
\item $\forall m \geq 5, \gamma_2(3,m) = \gamma_2(3,m-3) + 4$;
\item $\forall m \geq 8, \gamma_2(4,m) = \gamma_2(4,m-4) + 7$;
\item $\forall m \geq 14, \gamma_2(5,m) = \gamma_2(5,m-7) + 15$;
\item $\forall m \geq 20, \gamma_2(6,m) = \gamma_2(6,m-11) + 28$;
\item $\forall m \geq 31, \gamma_2(7,m) = \gamma_2(7,m-18) + 53$;
\item $\forall m \geq 16, \gamma_2(8,m) = \gamma_2(8,m-3) + 10$;
\item $\forall m \geq 17, \gamma_2(9,m) = \gamma_2(9,m-3) + 11$;
\item $\forall m \geq 14, \gamma_2(10,m) = \gamma_2(10,m-1) + 4$;
\item $\forall m \geq 16, \gamma_2(11,m) = \gamma_2(11,m-3) + 13$;
\item $\forall m \geq 17, \gamma_2(12,m) = \gamma_2(12,m-3) + 14$.
\end{itemize}
\end{multicols}

Thanks to these relations, and to the first values we obtain for each $n$, we deduce the formulas for $\gamma_2(n,m)$, for $1 \leq n \leq 12$. For instance, for $n = 5$ we only need to know the recurrence relation, plus the first twelve values. We stop at $n=12$ here because the method for arbitrarily large $n$ works for $n \geq 13$.

\subsection{Computing 2-domination numbers for arbitrarily large $n$}
We adapt here the method Gonçalves et al. introduced in \cite{rao}. Their idea was to assume that for a sufficiently large grid, the position of the stones would be a projection of an optimal tiling for $\ZZ^2$, except on a fixed-height border of the grid, because not every cell has 4 neighbours at the frontier of a finite grid. This happens to be also the case for the 2-domination problem.

We use the dual concept of \emph{loss} to count how many stones we need to make up for the degree problem of the border. The loss denotes how much "influence" given by the stones of a 2-dominating set was wasted. For instance, two neighbouring stones would cause a loss of 2: each stone cell dominates its stone neighbour which did not need to be dominated by another stone. Instead of computing the minimum number of stones needed, we will compute a lower bound on the minimum loss possible on the border. It happens that this lower bound gives us a direct lower bound on the 2-domination number, and that these bounds are sharp.\\

More formally, given a 2-dominating set $D$ of the $n\times m$-grid, we define the \emph{loss} to be $\ell(D,n,m) = 4|D|-2(nm-|D|)$. The idea behind this formula is simple: each stone contributes to the domination of its four neighbours, and each cell not in $D$ should be dominated twice. The difference between these two quantities is the influence of stones that was "lost", \emph{i.e.} not necessary. The loss function should have several characteristics: while it should be "easy" to compute, it should also be reversible so that given the loss, we can find the 2-domination number.

Here we can indeed reverse the formula to get $|D| = (2nm+\ell(D,n,m))/6$. Let $\ell(n,m)$ be the minimum possible loss over every 2-dominating set $D$. We then obtain $\gamma_2(n,m) = (2nm+\ell(n,m))/6$. The method works: the lower bound we obtain matches the 2-domination numbers (we prove later that it is also an upper bound). Computing the minimum loss over a big grid seems very hard, but we managed to find a lower bound for $\ell(n,m)$ by computing the minimum loss on a fixed-height border of the grid (see \autoref{figure-loss}).

\begin{figure}[h]
\centering
\includegraphics{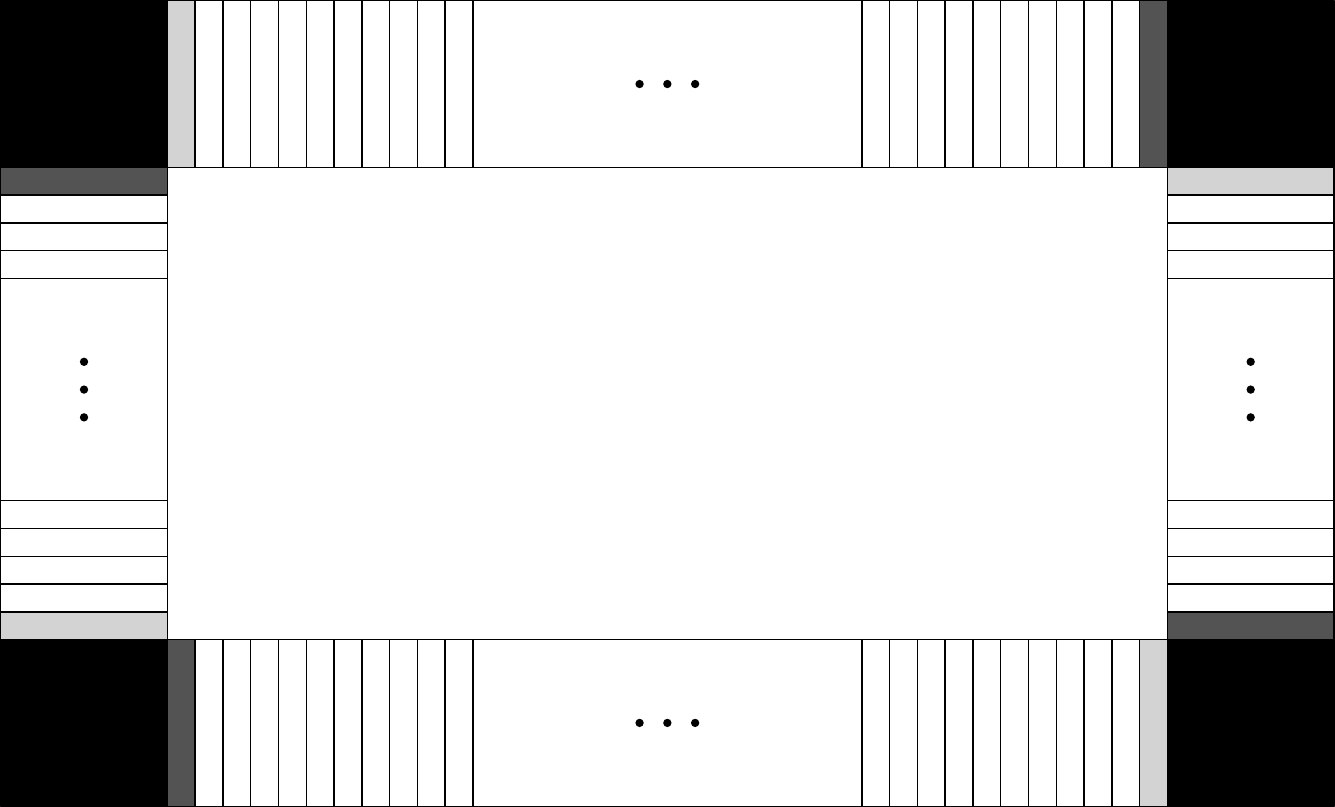}
\caption{The borders of the grid. The four parts coloured in black are the \emph{corners}, and the white parts are the \emph{bands}. The grey cells belong to both the bands and the corners: the ones filled with light grey are the output of a band and input of the corner next to it; the ones in dark grey are the output of a corner and input of the following band.}
\label{figure-loss}
\end{figure}

Computing the loss on a border can be done by adapting the method used to compute the 2-domination number when the number of lines is small. We compute a lower bound on the loss by computing only the loss over the border of the grid. Let $h > 0$ be an integer and an $n \times m$ grid with $n > 2h$ and $m > 2h$. The \emph{border} of height $h$ is the set of cells $(i,j)$ such that either $\min(i, n-1-i) < h$ or $\min(j, m-1-j) < h$. We define the \emph{corners} as the four connected parts of the grid composed of cells $(i,j)$ such that both $\min(i, n-1-i) < h$ and $\min(j, m-1-j) < h$. The remaining four connected parts of the border are called the \emph{bands}. The corners and bands are illustrated in \autoref{figure-loss}. Once again, we get an algorithm which is faster than exhaustive search over the dominating sets by working with states. However, we need to adapt the sets $\VV, \FF, \DD$ and the relation $\RR$ we worked with.

Let us begin with the bottom band. In this subsection, we focus on the bottom $h$ lines of our grid by adapting the sets $\VV, \DD$ and the relation $\RR$. Let assume that $n > h$ and consider the function $\hat{f_i}$
such that if $D$ is a 2-dominating set, $\hat{f_i}(D) = f_i(D)[0], \cdots f_i(D)[h-1]$. $\hat{f_i}$ consists of the bottom $h$ lines of $f_i$.
As in the previous subsection, $\hat{f}_i$ denotes the column $i$ of $\hat{f}$. We begin by defining the \emph{almost valid} states: $\VV_a = \cup_{0\leq i < m}{\{\hat{f}_i(D) : D \textrm{ is a 2-dominating set}\}}$. More explicitely, $S \in \VV_a$ if and only if, for $i \in \{0, \ldots, h-1\}$:
\begin{itemize}[noitemsep, topsep=0pt]
\item if $S[i] = $\textsc{need\_one} then at most one among $S[i-1]$, $S[i+1]$ is \textsc{stone};
\item if $S[i] = $\textsc{ok} then $i \neq 0$ or at least one among $S[i-1]$ and $S[i+1]$ is \textsc{stone};
\end{itemize}
\vspace{5pt}
Notice that there is a distinction depending on whether $i = 0$: the first cell will have a neighbour in the center of the grid, so we need to consider the case when this neighbour is a stone.
We do not need first states here, so we will not define a set $\FF_a$. We define the relation of \emph{almost-compatibility}: if $S, S'\in \VV_a$, $S \RR_a S'$ if and only if for $i \in \{0, \ldots, n-1\}$:
\begin{itemize}[noitemsep, topsep=0pt]
\item if $S[i] = $\textsc{need\_one} then $i = 0$ or $S'[i] = $\textsc{stone};
\item if $S'[i] = $\textsc{need\_one} and $i \neq 0$ then exactly one among $S'[i-1]$, $S'[i+1]$ and $S[i]$ is \textsc{stone};
\item if $S'[0] = $\textsc{need\_one} then at most one among $S'[i-1]$, $S'[i+1]$ and $S[i]$ is \textsc{stone};
\item if $S'[i] = $\textsc{ok} and $i \neq 0$ then at least two among $S'[i-1]$, $S'[i+1]$ and $S[i]$ are \textsc{stone};
\item if $S'[0] = $\textsc{ok} then at least one among $S'[i-1]$, $S'[i+1]$ and $S[i]$ is \textsc{stone}.
\end{itemize}
\vspace{5pt}
Finally, a set $S \in \VV_a$ is \emph{almost-two-dominated} (denoted by $S\in \DD_a$) if all its cell except the upper one are different from \textsc{need\_one}.

We again use the exponentiation of a transfer matrix to compute the minimum loss over a border of a grid. We define the matrix $T_a$ such that $T_a[S][S']$ contains the loss induced by putting state $S'$ after state $S$. By exponentiating the matrix $T_a$ we can compute the minimum loss over a border, excluding the loss induced by the first state alone on itself.

The next step is to compute the loss for corners. A corner is composed of an $h$ by $h$ square, plus an input column and an output column. Let us consider the bottom right corner of \autoref{figure-loss}. The last column of the bottom band is coloured in light grey: it is the input column of the square (and the output column of the band). At the other side of the square, the horizontal "column" filled with dark grey is the output column of the square (and the input column of the next border). Suppose that the input column of the square is in state $A$ and its output column is in state $B$. The loss over the corner is the sum of:
\begin{itemize}[noitemsep, topsep=0pt]
	\item the loss on the corner by $A, B$ and the corner itself;
	\item the loss on $B$ by the corner and $B$ itself;
	\item the loss on $A$ by the corner.
\end{itemize}
\vspace{5pt}
The explanation is simple: the input state was fixed by the loss computation on the band (so its loss so far was already counted) and the corner provides the first state for the next band (so we have to compute its loss so far).
Similarly to $T_a$, we define the matrix $C_a$ for a corner: $C_a[S][S']$ contains the minimum loss over a corner whose input state is $S$ and output state is $S'$ as defined just before (we do not count the loss induced by $S$ alone on itself).

\begin{claim}
$\min_{S \in \VV_a}((T_a^{m-2h-1}C_aT_a^{n-2h-1}C_a)^2[S][S])$ is the minimum loss over the border.
\end{claim}

\begin{proof}
$T_a^{m-2h-1}$ is the minimum loss over a band starting on the output column of the bottom left corner and ending on the input column of the bottom right corner. Hence $T_a^{m-2h-1}C_a$ means computing the minimum loss on the bottom band we have just described, and extending it to the output state of the bottom right corner. As mentioned above, in the corner loss we take the input state as it is (which is exactly what $T^{m-2h-1}$ provides: the loss on the last state by itself and its preceding column was already computed). Since we compute the loss on the output state of the corner, $T_a^{m-2h-1}C_aT_a^{n-2h-1}$ extends the loss to the right band. Now, $T_a^{m-2h-1}C_aT_a^{n-2h-1}C_a$ corresponds to the loss from the output of the bottom left square to the output of the top right square, that is the loss of half the border. By squaring this matrix, we obtain the minimum losses over the whole border of the grid: $(T_a^{m-2h-1}C_aT_a^{n-2h-1})^2[S][S]$ means that we compute the minimum loss by beginning from the $(h+1)^{th}$ column at the bottom in state $S$ and leaving it in state $S$ by the bottom left corner.
\end{proof}

In the rest of this section, we consider that $h = 6$. This value is sufficient to obtain the correct bounds with our program. Here again, we have the problem of computing the minimum loss over borders of arbitrary sizes. However, we may notice that, if we let $H(n,m) = (T_a^{m-13}C_aT_a^{n-13})^2$, there exist some $j_0, k$ and $p$ such that $\forall\; r \geq r_0, T_a^{r+k} = T_a^{r}+p$, so that $H(n+i, m+j) = H(n,m) + (i+j)p \;\;\;\forall\; n,m \geq 13+r_0$. Indeed, the matrix $T_a$ is primitive for the same reasons as for the transfer matrix of \autoref{claim-primitive}.

To complete the proof of \autoref{th-2dom}, we check the values of $\ell(n,m)$ for $13 \leq n \leq m \leq 35$. With these values, plus the recurrence relation on $T_a^{r}$, we achieve the proof of the theorem. Indeed, if $n > 13$ and $33 \leq m \leq 35$ then for all $k \in \NN$:
\begin{align*}
\gamma_2(n,m+3k) = \frac{2n(m+3k)+\ell(n,+3k)}{6} \geq &  \frac{2nm+\ell(n,m)}{6} + nk+2k \\\ \geq & \left\lfloor \frac{(n+2)(m+2/)}{3}-6 \right\rfloor +nk +2k \\ \geq & \left\lfloor \frac{(n+2)(m+3k+2)}{3}-6 \right\rfloor.
\end{align*}
This proves the lower bound for every $13 < n \leq 35$ and $m \in \NN$. To prove it for $n > 35$, it suffices to do the exact same computation, to compute $\gamma_2(n+3k, m)$ for any $m$.
\begin{figure}[h!]
\centering
\includegraphics[scale=0.75]{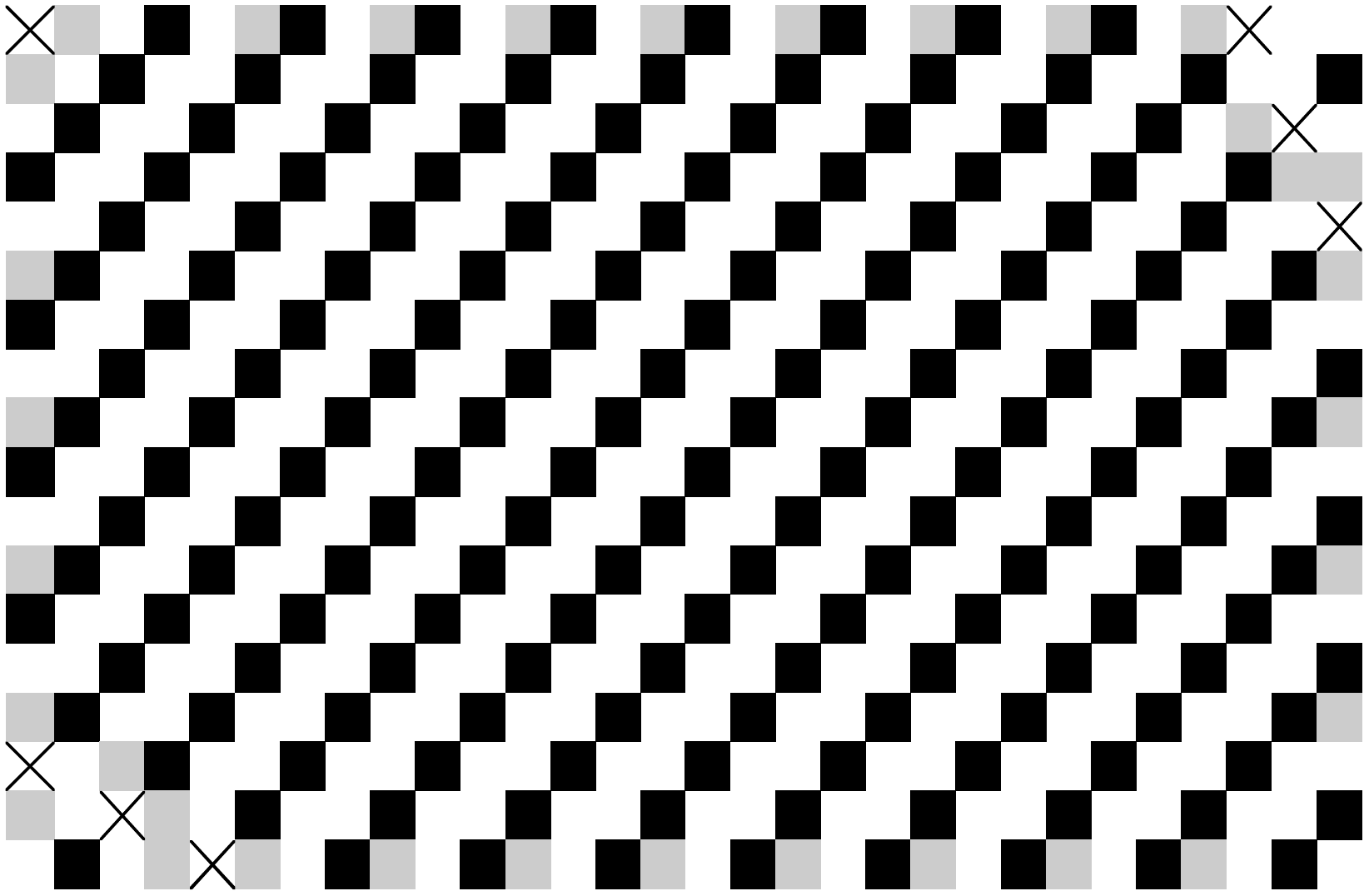}
\caption{Example of an optimal 2-dominating set $D$ on a $18 \!\times\! 30$ grid. $D$ is the set of cells which are gray or black. The black cells and the cells with a cross are the projection of a minimal $\ZZ^2$ 2-dominating set on the grid.}
\label{ex-dominating-set}
\end{figure}

To show that this bound is sharp, we give general 2-dominating sets of the right size. To obtain these 2-dominating sets, we select, for an infinite grid $\ZZ^2$ the 2-dominating set $D = \{(i,j) : i+j \mod 3 = 0\}$ and its rotations. We then take all the different restrictions of these 2-dominating sets for $\ZZ^2$ into a finite $n \!\times\! m$ grid. For each restriction we obtain, we modify each corner of size 6 according to two rules which depend on the pattern in that corner. The two rules are shown in \autoref{ex-dominating-set}: the first rule is used in the top-left corner, and the second rule is used in the top-right and bottom-left corners. A rule corresponds to removing the cells with a cross and adding the grey cells of the corner. Rule 1 could be stated as follows: if the cell at the angle of the grid is in the dominated set, we remove it from the set and add instead its two neighbours. Finally, we put a stone on the cells of the first and last rows and columns which are not 2-dominated. One can show that for $14 \leq n \leq m$ one of the resulting 2-dominating $D_{n,m}$ set has the right size. We can see an example of such a $D_{n,m}$ for a $18 \!\times\! 30$ grid in \autoref{ex-dominating-set}. The first rule is used in the top left corner and the second rule is used in the top right and bottom left corner. No modification needs to be done in the bottom right corner. By counting the number of stones in the regular pattern (black and crossed cells in), removing the number of crossed cells, and adding the number of grey cells, we get $D = \frac{nm+2n+2m}{3}-5$, which is equal to the number in \autoref{th-2dom} when $n$ and $m$ are multiple of 3.

The grid is of size $18 \!\times\! 30$, but it extends immediately to any $n\!\times\! m$ grid when $n$ and $m$ are both greater than 14 and multiple of 3. Applying the same method for $14 \leq n,m$ when the two numbers have other congruences modulo 3 lead to 2-dominating sets having the right size.

\section{Application to the Roman domination problem}
In this section we consider another domination problem: the Roman domination. Formally, a Roman-dominating "set" is a pair $(S_1, S_2)$ such that every vertex $v\notin S_1\cup S_2$ has at least one neighbour in $S_2$. The cost of such a Roman-dominating set is $|S_1|+2|S_2|$. Intuitively, the set $S_1$ is the set of vertices on which we put one stone, and they are dominated by themselves. $S_2$ is the set of vertices on which we put two stones, which makes them dominated, and they dominate their neighbours. The name of the problem comes from the times when the Roman were conquerors: if they wanted to defend one of their regions they could put either one troop, so that the region is guarded, or put two troops (at twice the expense), so that the soldiers could also be used to defend any neighbouring region.

We will prove the following theorem:

\begin{thm}
For all $1 \leq n \leq m$,\\
\[\gamma_R(n,m) =    \left\{
\setstretch{1.5}
\begin{array}{ll}
      \left\lceil\frac{2m}{3}\right\rceil & \quad\textrm{if }n = 1 \\
      m+1 & \quad\textrm{if }n=2 \\
      \left\lceil \frac{3m}{2} \right\rceil& \quad\textrm{if } n=3\textrm{ and } m\mod 4 = 1 \\
      \left\lceil \frac{3m}{2} \right\rceil+1& \quad\textrm{if } n=3\textrm{ and } m\mod 4 \neq 1 \\
      2m+1& \quad\textrm{if } n=4\textrm{ and } m = 5 \\
      2m& \quad\textrm{if } n=4\textrm{ and } m > 5 \\
      \left\lfloor\frac{12m}{5}\right\rfloor+2& \quad\textrm{if } n=5\\
      \left\lfloor\frac{14m}{5}\right\rfloor+2& \quad\textrm{if } n=6\textrm{ and } m\mod 5 \in \{0,3,4\}\\
      \left\lfloor\frac{14m}{5}\right\rfloor+3& \quad\textrm{if } n=6\textrm{ and } m\mod 5 \notin \{0,3,4\}\\
      \left\lfloor\frac{16m}{5}\right\rfloor+2& \quad\textrm{if } n=7 \textrm{ and } m = 7 \textrm{ or }m\mod 5 = 0\\
      \left\lfloor\frac{16m}{5}\right\rfloor+3& \quad\textrm{if } n=7 \textrm{ and } (m > 7\textrm{ and }m\mod 5 \neq 0)\\
      \left\lfloor\frac{18m}{5}\right\rfloor+4& \quad\textrm{if } n=8 \textrm{ and } m \mod 5 = 3\\
      \left\lfloor\frac{18m}{5}\right\rfloor+3& \quad\textrm{if } n=8 \textrm{ and } m \mod 5 \neq 3\\
      \left\lfloor\frac{20m}{5}\right\rfloor+2& \quad\textrm{if } n=9 \textrm{ and } m \mod 5 = 4\\
      \left\lfloor\frac{20m}{5}\right\rfloor+3& \quad\textrm{if } n=9 \textrm{ and } m \mod 5 = 4\\
      
      \left\lfloor \frac{2(n+1)(m+1)-2}{5}\right\rfloor-1 & \quad\textrm{if } n\geq 10 \textrm{ and } n\mod 5 = 4 \textrm{ and } m \mod 5 = 4\\
      \left\lfloor \frac{2(n+1)(m+1)-2}{5}\right\rfloor & \quad\textrm{if } n\geq 10 \textrm{ and } n\mod 5 \neq 4 \textrm{ or } m \mod 5 \neq 4
\end{array}
\right. \]
\label{th-Roman-domination}
\end{thm}

Since the rules of the Roman domination are a bit different, and we can put two stones on a cell, we need a slight adaptation of the states and the loss. The possible states for a cell are now: \textsc{two\_stones}, \textsc{stone}, \textsc{ok} and \textsc{need\_one}.
A state $S$ is in $\VV$ if and only if for $i \in \{0, \ldots, n-1\}$:
\begin{itemize}[noitemsep, topsep=0pt]
\item if $S[i] = $\textsc{need\_one} then neither $S[i-1]$ nor $S[i+1]$ is \textsc{two\_stones};
\item if $S[i] = $\textsc{stone} then neither $S[i-1]$ nor $S[i+1]$ is \textsc{two\_stones} or \textsc{stone};
\end{itemize}
\vspace{5pt}
The second rule is not required for the coherency of the state, but it is an optimisation which allows us to reduce a lot the number of states. It is justified by the fact that in a minimum Roman dominating set, we can always remove any stone neighbouring a cell with two stones, and if there are two neighbouring cells with a stone each we still have a dominating set of same value by removing one of the stones and putting a second stone on the other cell.

A state $S \in \VV$ is in $\FF$ if and only if for every $i \in \{0, \ldots, n-1\}$, if $S[i] = $\textsc{ok} then at least one among $S[i-1]$ and $S[i+1]$ is \textsc{two\_stones}.

$(S, S')$ is a compatible pair if and only if for $i \in \{0, \ldots, n-1\}$:
\begin{itemize}[noitemsep, topsep=0pt]
\item if $S[i] = $\textsc{need\_one} then $S'[i] = $\textsc{two\_stones};
\item if $S'[i] = $\textsc{need\_one} then $S[i] \neq $\textsc{two\_stones};
\item if $S'[i] = $\textsc{ok} then at least one among $S[i]$, $S'[i-1]$ and $S'[i+1]$ is \textsc{two\_stones};
\item if $S[i] \in \{\text{\textsc{two\_stones}, \textsc{stone}}\}$ then $S'[i] \neq $\textsc{stone};
\item if $S[i] = $\textsc{stone} then $S'[i] \neq $\textsc{two\_stones}.
\end{itemize}
\vspace{5pt}
Finally, a state $S \in \VV$ is in $\DD$ if and only if none of its entry is \textsc{need\_one}.\\

We now need to adapt the loss. Here we define $\ell(n,m) = 5|S_2|+5/2|S_1|-nm = (2|S_2|+|S_1|)2/5 -nm$. Indeed, each cell with two stones dominates 5 cells, and each cell in $S_1$ dominates only itself, but we add to it an additional loss of $3/2$ to penalize its bad ratio of number of dominated cells compared to number of stones used. This allows us to get $\gamma_R(n,m) \geq (\ell(n,m)+nm)5/2$. Note that in the program, what we compute is actually $2\ell(n,m)$ to avoid to manipulate fractions or floating numbers.
Let us define the almost-valid states which, for this problem, coincide with the valid states: $\VV_a = \VV$.
Now if $S, S' \in \VV_a$, $S \RR_a S'$ if and only if for $i \in \{0, \ldots, n-1\}$:
\begin{itemize}[noitemsep, topsep=0pt]
\item if $S[i] = $\textsc{need\_one} and $i \neq 0$ then $S'[i] = $\textsc{two\_stones};
\item if $S'[i] = $\textsc{need\_one} then $S[i] \neq $\textsc{two\_stones};
\item if $S'[i] = $\textsc{ok} and $i \neq 0$ then at least one among $S[i]$, $S'[i-1]$ and $S'[i+1]$ is \textsc{two\_stones};
\item if $S[i] \in \{\text{\textsc{two\_stones}, \textsc{stone}}\}$ then $S'[i] \neq $\textsc{stone};
\item if $S[i] = $\textsc{stone} then $S'[i] \neq $\textsc{two\_stones}. 
\end{itemize}
\vspace{5pt}
Here again we do not give complete details on how we compute the loss. Since we compute twice the loss, each cell with two stones having $k < 4$ neighbours contributes for $4-k$, and each cell dominated by $k > 1$ cells also contributes for $k-1$. Finally, each stone with one cell contributes for $3/2$. All these contributions sum up to make the loss. We recall that in the program we compute twice these values.\\

As in the previous section, we get exact values for "small" values of $n$, and a lower bound for bigger values of $n$. In Chapter 4 of the thesis of Currò \cite{curro}, the Grid Theorem gives an upper bound which matches our lower bound, hence this is the exact value.

\section{Concluding remarks and open problems}
We successfully adapted the techniques introduced in \cite{rao} to other dominating problems, namely the 2-domination and the Roman domination problems. The techniques used in this paper could be reused to show similar results on grids, that is Cartesian products of paths.

Some authors investigated the problem of domination in Cartesian products of cycles (see for instance \cite{kla,pav}). The first part of the technique (when $n$ is fixed and small) may be adapted (with some care) but the second part (for arbitrary number of lines) does not apply directly since a crucial property is that the loss can be concentrated inside the \emph{borders} of the grids.

We mentioned in the introduction that the fact that the method gives sharp bounds are probably related to some tiling properties.
In the case of the 2-domination and the Roman domination, it is not properly speaking a tiling problem, but a generalised tiling problem with some weights (see \autoref{2-roman-figures}). The properties we write below are rather focused on standard tilings.
\begin{figure}[h]
\hspace{-0.35cm}
\begin{subfigure}[The 2-domination shape]
{\hspace{-2cm}\includegraphics[scale=0.6]{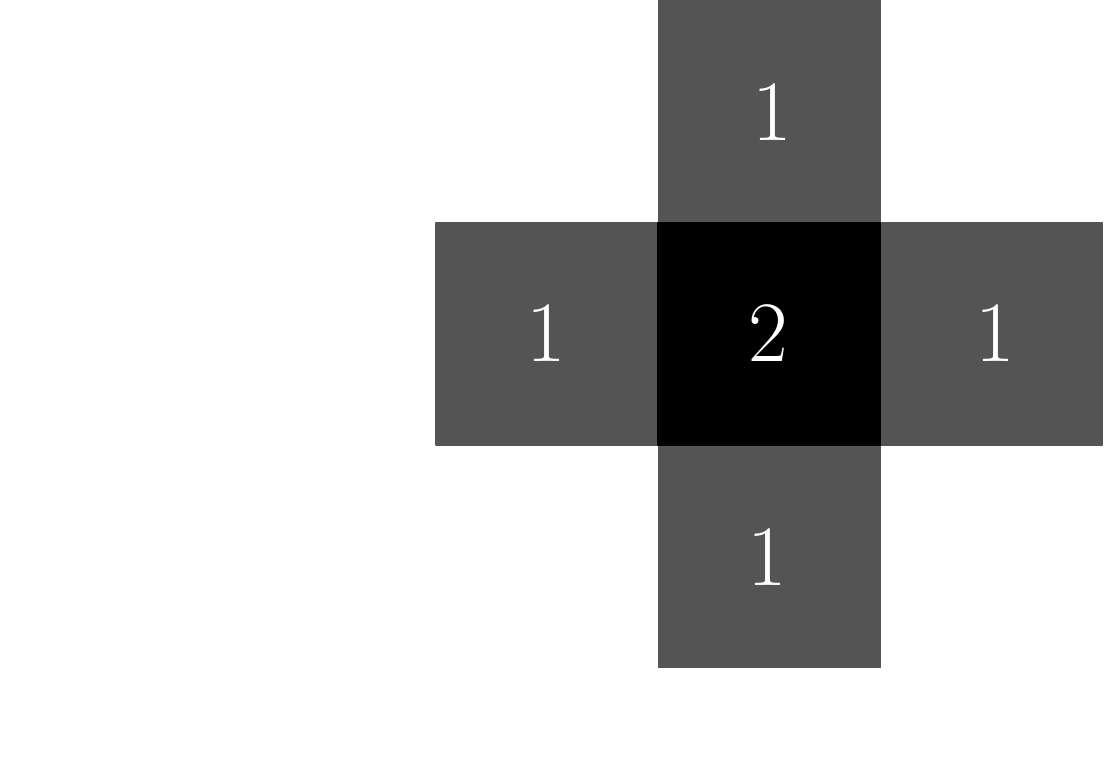}}
\end{subfigure}\hspace{2.55cm}
\begin{subfigure}[The Roman domination shapes]
{\includegraphics[scale=0.6]{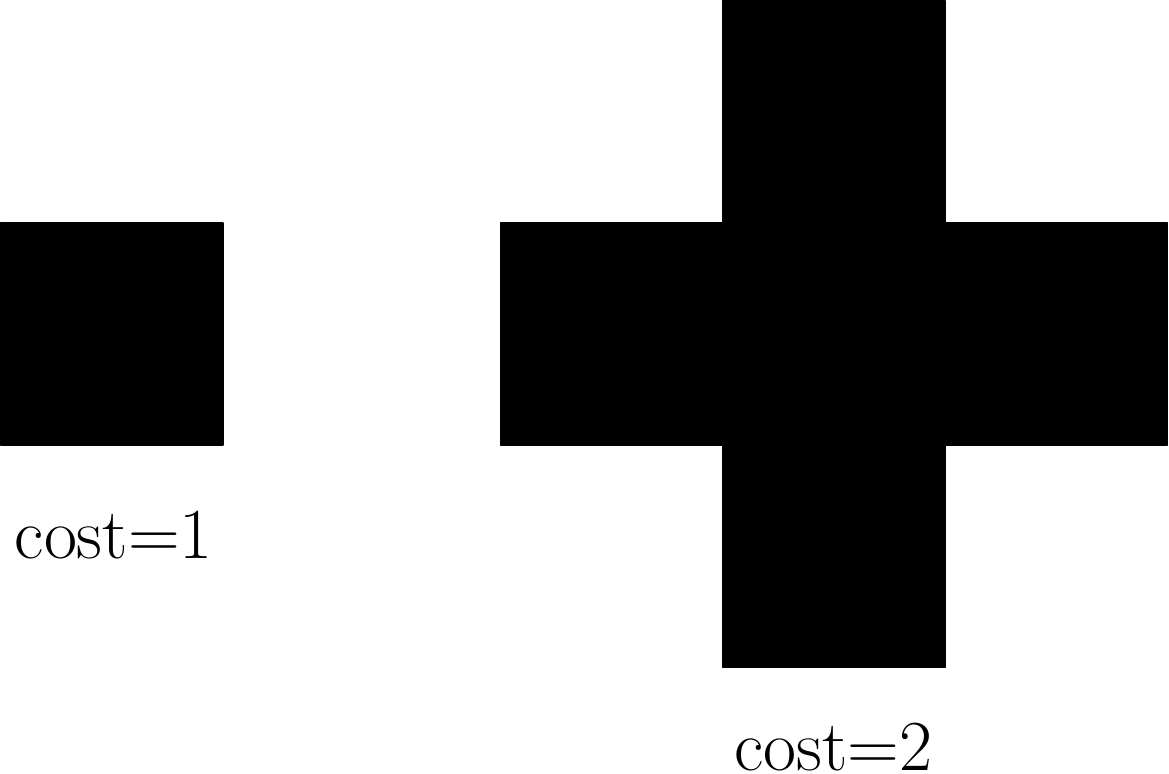}}
\end{subfigure}

\caption{The shapes for the 2-domination and the Roman domination. For the 2-domination, we look for a covering such that the sum of the weights (in white) on a cell is at least 2. For the Roman domination, we cover with two tiles, but they have different costs. We are interested in a covering of minimum weight.}
\label{2-roman-figures}
\end{figure}

The dynamic algorithm or transfer matrix exponentiation for grids of small heights is likely to work for any similar problem. Indeed, for a fixed number of lines and columns, it only needs an adaptation of the special sets of states and of the compatibility relation. To extend it to an infinite number of columns, it is sufficient to have a primitive transfer matrix.
However, properties enabling the method to work for a (arbitrarily) large number of lines are yet to be found. One crucial point is the following property.
\begin{ppty}[Fixed-height border-fixing]
Let $X$ be a shape. $X$ has the fixed-height border-fixing property if there exist $k, n_0, m_0$ such that, for any $n \geq n_0$ and $m \geq m_0$, there exists an optimal \underline{covering} of the $n \times m$ rectangle whose cells at distance greater than $k$ of the border constitute a subtiling of the plane.
\end{ppty}
For instance, the 2-domination shape has this property for $k=3$: any optimal solution to the 2-domination problem can be obtained from an infinite optimal 2-domination set of which we modify only cells at distance at most 3 from the border. Note that, due to the automation feature of the algorithm, this is indeed $k=3$ here even if the program needs to explore borders of size 6 to find the correct bounds.

The fixed-height border-fixing property implies that the bounds given by the method are sharp for some height of band, independent of the size of the rectangle. It seems to be related to the following property.

\begin{ppty}[Crystallisation]
Let $X$ be a shape. We say that $X$ has the crystallisation property if there exists $k \in \NN$ such that for every partial tiling of size $k$ with the shape $X$, either this tiling cannot be extended to tile the plane, or there is a unique way to do so.
\end{ppty}

For instance, the domination shape has this property for $k=2$. On the contrary, the total domination, which has been studied in grids by Gravier \cite{gravier} does not have this property.
The total-domination problem has been studied a lot in other graphs (see \cite{review} for example), but remains open for grids.
The total-domination problem is related to the shapes in \autoref{total-shape-fig}. The small one corresponds to the influence of one "stone": note that it does not dominates itself. The big ones are the unions of two copies of the small one. One can see that tiling the plane with the small shape is equivalent to tiling the plane with the set of the two big shapes: in the small shape, the cell middle cell must be dominated. As shown, the big shape can be vertical or horizontal. The problem with our technique is that a tiling of the plane can, with a certain degree of freedom, mix the vertical and the horizontal big shapes. This probably leads to some non-zero loss in the center of a big grid to be necessary for a covering to be of minimum size. In this case the assumption of the loss on the border being zero would be false, making our technique not usable.

\begin{figure}
\centering
\includegraphics[scale=1]{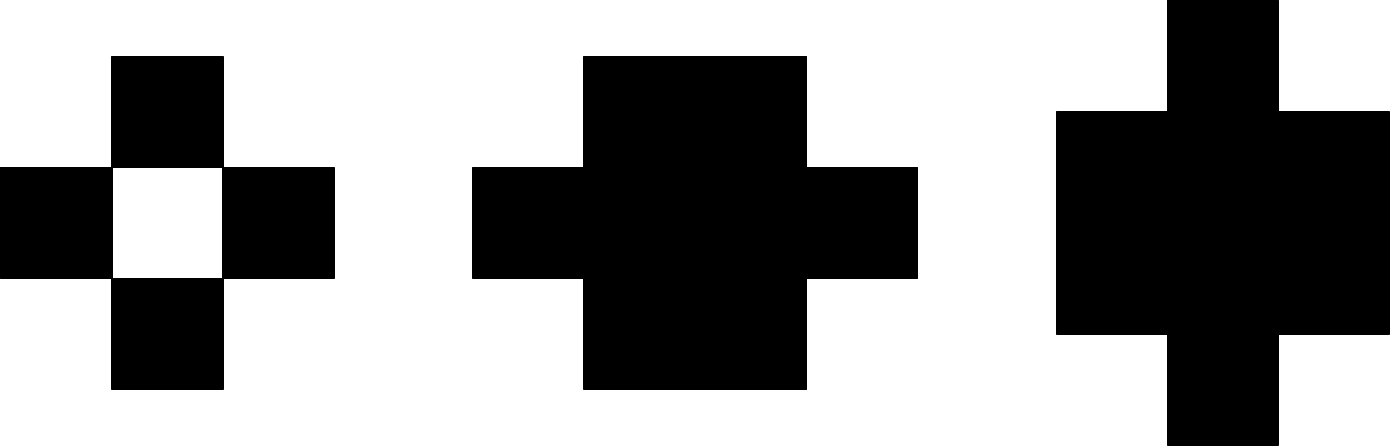}
\caption{The shapes associated with the total domination. The big ones are the two different unions of two copies of the small one. Tiling the plane with the small one boils down to tiling the planes with the two big ones.}
\label{total-shape-fig}
\end{figure}

\begin{conj}
If a shape $X$ tile the plane and has the crystallisation property then it also has the fixed-height border-fixing property.
\end{conj}

These properties could also be used on covering problems even if they have no relation with any domination problem on grids.

\end{document}